\documentclass[11pt]{article}
\usepackage[letterpaper,margin=1.00in]{geometry}
\usepackage{amsmath, amssymb, amsthm, amsfonts}
\usepackage{bbm}

\usepackage{ifthen}
\usepackage{tikz}
\usetikzlibrary{positioning,decorations.pathreplacing}

\usepackage{cite}
\usepackage{appendix}
\usepackage{graphicx}
\usepackage{color}
\usepackage{color}
\usepackage{algorithm}
\usepackage{algorithm}
\usepackage[noend]{algpseudocode}
\usepackage{epstopdf}
\usepackage{wrapfig}
\usepackage{paralist}
\usepackage{wasysym}
\usepackage[textsize=tiny]{todonotes}

\usepackage{framed}
\usepackage[framemethod=tikz]{mdframed}
\usepackage[bottom]{footmisc}
\usepackage{enumitem}
\setitemize{noitemsep,topsep=3pt,parsep=3pt,partopsep=3pt}
\usepackage[font=small]{caption}
\usepackage{xspace}

\newtheorem{theorem}{Theorem}[section]
\newtheorem{lemma}[theorem]{Lemma}
\newtheorem{meta-theorem}[theorem]{Meta-Theorem}
\newtheorem{claim}[theorem]{Claim}

\newtheorem{corollary}[theorem]{Corollary}

\definecolor{darkgreen}{rgb}{0,0.5,0}
\usepackage{hyperref}
\hypersetup{
    unicode=false,          
    colorlinks=true,        
    linkcolor=red,          
    citecolor=darkgreen,        
    filecolor=magenta,      
    urlcolor=cyan           
}
\usepackage[capitalize, nameinlink]{cleveref}
\Crefname{lemma}{Lemma}{Lemmas}
\Crefname{claim}{Claim}{Claims}
\Crefname{remark}{Remark}{Remarks}
\Crefname{observation}{Observation}{Observations}
\algnewcommand\algorithmicswitch{\textbf{switch}}
\algnewcommand\algorithmiccase{\textbf{case}}

\algdef{SE}[SWITCH]{Switch}{EndSwitch}[1]{\algorithmicswitch\ #1\ \algorithmicdo}{\algorithmicend\ \algorithmicswitch}%
\algdef{SE}[CASE]{Case}{EndCase}[1]{\algorithmiccase\ #1}{\algorithmicend\ \algorithmiccase}%
\algtext*{EndSwitch}%
\algtext*{EndCase}%

\newcommand{\eps}{\varepsilon}

\newcommand{\local}{$\mathsf{LOCAL}$\xspace}

\newcommand{\poly}{\operatorname{\text{{\rm poly}}}}

\renewcommand{\paragraph}[1]{\vspace{0.15cm}\noindent {\bf #1.}}

\newcommand{\FullOrShort}{full}

\ifthenelse{\equal{\FullOrShort}{full}}{
	
  \newcommand{\fullOnly}[1]{#1}
  \newcommand{\shortOnly}[1]{}
}{
    \newcommand{\shortOnly}[1]{#1}
    \newcommand{\fullOnly}[1]{}
  }

\begin{document}
\date{}
\title{Local Computation of Maximal Independent Set}

\author{
  Mohsen Ghaffari\thanks{Supported in part by funding from the European Research Council (ERC), under the European Union’s Horizon 2020 research and innovation program (grant agreement 853109).
  }\\
  \small MIT \\
  \small ghaffari@mit.edu
 }

\maketitle

\begin{abstract}
We present a randomized Local Computation Algorithm (LCA) with query complexity $\poly(\Delta) \cdot \log n$ for the Maximal Independent Set (MIS) problem. That is, the algorithm  determines whether each node is in the computed MIS or not using $\poly(\Delta) \cdot \log n$ queries to the adjacency lists of the graph, with high probability, and this can be done for different nodes simultaneously and independently. Here $\Delta$ and $n$ denote the maximum degree and the number of nodes. This algorithm resolves a key open problem in the study of local computations and sublinear algorithms (attributed to Rubinfeld).  
\end{abstract}

\setcounter{page}{0}
\thispagestyle{empty}

\vspace{1cm}

\newpage
\section{Introduction and Related Work}
When dealing with massive data, even linear time algorithms might be too slow. A natural and useful paradigm in the area of sublinear algorithms, especially when dealing with problems where the output may be large (and thus computing or even storing the entire output may be infeasible), is that of \emph{local computations}. Here, the algorithm should be able to compute any particular part of the output in sublinear time. In this paper, we study local computation algorithms for the well-studied Maximal Independent Set (MIS) problem. Before describing our contribution, let us review the background of the problem, including the model and the known results.

\subsection{Background: Prior Work and Model}
We review the prior work on local computations of MIS in two categories, after recalling the formal model. We note that the results in the first category were obtained before the introduction of the formal model (and technically, provide a weaker expected-per-vertex query complexity guarantee).

\paragraph{Local Computation Algorithms} As introduced by Rubinfeld et al.~\cite{rubinfeld2011LCA} and Alon et al.~\cite{alon2012LCA}, a Local Computation Algorithm (LCA) has query access to the graph where each query $(v, i)$ returns the $i^{th}$ neighbor of vertex $v$ (or $\emptyset$, if there is no such neighbor). The algorithm also has access to a string of randomness. In the case of the MIS problem, when \emph{questioned} about a vertex $v$, the algorithm performs a number of queries to the graph (and the string of randomness) and outputs a YES/NO answer indicating whether $v$ is in the MIS or not. These outputs are independent of the questions to the algorithm, and indeed, one can question the algorithm about different vertices simultaneously and independently. The measure of complexity is the worst-case number of queries the algorithm performs to answer any single question. We use $n$ and $\Delta$ to denote the number of vertices and the maximum degree, and the focus is generally on graphs where $\Delta \ll n$. See \cite{rubinfeld2011LCA, alon2012LCA} for related work, motivations, and applications, and \cite{levi2017centralized} for a survey. See also Lovasz's book on large networks and graph limits\cite[Section 22.3]{lovasz2012large}, which discusses such local computations and connections to distributed local algorithms.

\medskip
\paragraph{Category I (Expected Query Complexity Per Vertex)} Nguyen and Onak~\cite{nguyen2008constant}, as the core technical ingredient of their celebrated sublinear time approximation algorithms work, developed the first local algorithm for MIS. Their algorithm simulates the Randomized Greedy MIS (RGMIS) algorithm, in the sense that the YES/NO answers to the questions of whether a vertex $v$ is in the MIS or not are the same as those of one execution of RGMIS. In RGMIS, we process the vertices according to a random order $\pi$ and add each vertex to the MIS if and only if none of its neighbors that appeared earlier in the order $\pi$ have been added to the MIS. Nguyen and Onak's local algorithm answers each question in expected $2^{O(\Delta)}$ queries to the graph, where the expectation is over the randomness of the order $\pi$.  

Nguyen and Onak used this approach as a local computation for maximal matching\footnote{Maximal matching is a special (and simpler) case of the maximal independent set problem, on line graphs.}. Through this, they directly obtained a sublinear-time approximation algorithm for the size of maximal matching (and some other related problems such as approximate minimum vertex cover), by invoking the local computation on a small number of randomly selected vertices and determining what fraction of them are matched. It is worth noting that this itself was inspired by, and an improvement on, the insightful work of Parnas and Ron~\cite{parnas2007approximating}, who had developed a sublinear time approximation algorithm for the size of minimum vertex cover by devising a local computation procedure for it (via simulating a local distributed algorithm of that problem).

Yoshida, Yamamoto, and Ito~\cite{yoshida2009improved} provided an ingenious analysis for a variant of the approach of Nguyen and Onak, which improved the bound to expected $O(\Delta)$, though in a certain weaker sense: for a random permutation $\pi$ used in RGMIS and for a \underline{random vertex} $v$, the expected query complexity to determine whether $v$ is in the MIS or not is $O(\Delta)$. This bound holds also if one replaces $\Delta$ with average degree $\bar{\Delta}$. However, for some vertices, even under a random permutation, the expected number of queries to answer the question might be much larger. This weaker guarantee was still sufficient for the sublinear-time approximation application because there one deals with questions about random vertices. See \cite{onak2012near, behnezhad2021sublinear} for other improvements on these sublinear-time size approximation algorithms.

\bigskip
\paragraph{Category II (Query Complexity For All Vertices)} The algorithms discussed in the first category have bounds on the expected query complexity of each node (or just a random node). A stronger notion, which is indeed the prevalent measure of interest in the study of LCAs, is to have an upper bound on the query complexity of all nodes. This bound itself may hold in expectation, or even better, with high probability\footnote{We use the phrase \emph{with high probability (w.h.p.) to indicate that an event takes place with probability at least $1-1/n^{c}$ for a fixed constant $c\geq 2,$ which can be set desirably large by adjusting other constants.}}. Notice that the naive bound that follows directly from the per-vertex expected query complexity, via Markov's inequality and a union bound over all vertices, includes an $n$ factor, and is thus uninteresting in the area of sublinear algorithms\footnote{Here, unlike standard centralized computations, we cannot simply re-run the algorithm for $O(\log n)$ iterations, stopping each iteration after say twice the expected time, and use this to turn our expected complexity into a bound that holds with probability $1-1/n$. This is because, which of those $O(\log n)$ executions would be used is a non-local decision and depends on the questioned vertex, and if we do not fix one choice for all vertices, the output at different vertices could be inconsistent.}.

Rubinfeld et al.~\cite{rubinfeld2011LCA} and Alon et al.~\cite{alon2012LCA}, in their work that introduced the model, provided an MIS LCA with query complexity of $\Delta^{O(\Delta \log \Delta)} \poly(\log n)$. The algorithm is based on Luby's classic parallel/distributed MIS algorithm~\cite{luby1985simple}. This LCA also uses a random string of the same length (and no other memory, besides these). Concretely, when questioned about any vertex $v$, the algorithm performs only $\Delta^{O(\Delta \log \Delta)} \poly(\log n)$ queries to the input graph (and the randomness), with high probability. A natural question that has remained open since then, and was raised frequently by Rubinfeld (e.g., her 2016 TCS+ talk \href{https://youtu.be/R8J61RYaaDw?t=2164}{https://youtu.be/R8J61RYaaDw?t=2164}, her 2017 CICS Distinguished Lecture \href{https://youtu.be/fczTaR-KSd8}{https://youtu.be/fczTaR-KSd8}, and her Keynotes Talks at ICALP 2017 and PODC 2019) and others (see, e.g., \cite{mansour2012converting, levi2017local, ghaffari2019sparsifying, behnezhad2021sublinear}) is this: 
  
\begin{center}
\begin{minipage}{0.65\textwidth}
\emph{Is there an LCA for MIS with $\poly(\Delta \log n)$ query complexity?}
\end{minipage}
\end{center}
\smallskip

There has been a sequence of improvements. The complexity was improved to $2^{O(\Delta)} \poly(\log n)$ by Vardi and Reinhold~\cite{reingold2016new}, by showing that a variant of the Nguyen-Onak algorithm has this query complexity (for all vertices); that is, the query complexities are concentrated around the $2^{O(\Delta)}$ expected bound shown by Nguyn and Onak. Later, the bound was improved by devising new algorithms: Levi, Rubinfeld and Yodpinyanee\cite{levi2017local} achieved a query complexity of $\Delta^{O(\log^2 \Delta)} \poly(\log n)$; Ghaffari achieved a query complexity of $\Delta^{O(\log \Delta)} \poly(\log n)$; and finally Ghaffari and Uitto\cite{ghaffari2019sparsifying} improved the query complexity to $\Delta^{O(\log\log \Delta)} \poly(\log n)$. 

However, all these bounds remain quite far from the $\poly(\Delta \log n)$ target. Indeed, to the best of our knowledge, even the simpler question of whether one can obtain a local MIS algorithm that, for \underline{each} questioned vertex, has expected query complexity $\poly(\Delta)$---instead of ``for a \underline{random} vertex", as provided by Yoshida et al.~\cite{yoshida2009improved}---has remained open.

\paragraph{Side Remark} If we limit ourselves to graphs with constant or slightly super constant degrees, and focus instead on the dependency on $n$, then a better bound follows from the algorithm of Even, Medina, and Ron\cite{even2018best}, which has a query complexity of $\Delta^{O(\Delta^2)} \log^* n$, later improved to $\Delta^{O(\Delta)} \log^* n$\cite{levi2017centralized}. These bound are sublinear in $n$ only for graphs with $\Delta = o({\log n/\log\log n})$.

\subsection{Our Contribution}
We present a randomized local computation algorithm that resolves the aforementioned question and achieves a $\poly(\Delta \log n)$ query complexity. We also note that the expected query complexity for each questioned vertex is $\poly(\Delta)$.
\begin{theorem} There is a Local Computation Algorithm that computes an MIS such that, when questioned about any vertex $v$, it answers whether $v$ is in the MIS or not using expected $\poly(\Delta)$ queries, and at most $\poly(\Delta \log n)$ with probability at least $1-1/\poly(n)$. More concretely, for any $\eps>0$ and any vertex $v$, with probability at least $1-\eps$, the algorithm uses at most $\poly(\Delta) \cdot \log (1/\eps)$ queries to answer the question about $v$ (if asked). Furthermore, the algorithm can be transformed into one with $\poly(\Delta \log n)$ query complexity and using only $\poly(\Delta \log n)$ bits of randomness.
\end{theorem}

We remark that the constant in the exponent of our current bound $\poly(\Delta)$ is not small. We did not attempt to optimize that constant in this writing. We instead prioritized the simplicity of the algorithm and its analysis. It remains an interesting question whether one can achieve a better---e.g., linear or quadratic---dependency on $\Delta$.

MIS's centrality in this area is in part due to the fact that many other problems can be easily reduced to it. Using well-known reductions, we get local computation algorithms with the same asymptotic complexity for maximal matching, $2$-approximate vertex cover, $(\Delta+1)$ vertex coloring, and even its strengthening to $(\mathsf{deg}+1)$ list coloring where each node $v$ should choose its color from a prescribed list of $degree(v)+1$ many colors, and of course its special case of $(2\Delta-1)$ edge (list) coloring. 

\subsection{A High-Level Outline of Our Method}\label{subsec:outline}
Similar to many of the prior work in local computation algorithms, and following a connection first used by Paras and Ron\cite{parnas2007approximating} (see also Lovasz~\cite[Section 22.3]{lovasz2012large}), the core ingredient of our LCA is to simulate a local distributed algorithm for (nearly) maximal independent set. Let us first explain this connection. In a distributed algorithm, initially, nodes do not know the topology of the graph, and per round each node of the graph talks to all of its neighbors. Because of this, it can be shown that a node's behavior after $T$ rounds is a function of only the information initially available to the nodes within its distance $T$. Since the number of the latter nodes is at most $\Delta^{T}$, this indicates that any $T$-round local distributed algorithm can be turned into an LCA with query complexity $\Delta^{T}$. 

This connection, per se, is insufficient for obtaining a $\poly(\Delta)$ query LCA. By a celebrated lower bound of Kuhn, Moscibroda, and Wattenhofer\cite{kuhn2016local}, any distributed algorithm for MIS needs round complexity $\Omega(\log \Delta/\log\log \Delta)$. Hence, by only following the Parnas-Ron principle of directly simulating a distributed algorithm in a black-box fashion, we cannot obtain a query complexity below $\Delta^{\Omega(\log \Delta/\log\log \Delta)}$. 

To obtain our LCA, we use an adaptation of the local distributed algorithm of \cite{ghaffari2016MIS} to compute a near-maximal independent set (the adjustment to an MIS is easy and standard). We adapt the algorithm in such a way that, while maintaining its $T=\Theta(\log \Delta)$ round complexity, it has additional properties that allow us to simulate it using an LCA with query complexity of $\Delta\cdot 2^{\Theta(T)}$. We emphasize that this is a deterministic guarantee; the probabilistic aspects are in the near-maximality of the computed independent set. To describe the key new property, let us use an informal definition of \emph{(causal) influence}: let us say node $u$ is causally influenced by node $v$ if node $u$ reads a message of node $v$ or if $u$ reads a message of another node $w$ who is influenced by node $v$. The former is a direct influence and the latter is an indirect influence. The reason for the naive $\Delta^T$ bound in LCA simulations of $T$-round distributed local algorithms is that per round each node reads the messages of (up to) $\Delta$ neighbors, and these factors multiply over different rounds. To reach a query complexity of $\poly(\Delta)$, we adapt the local distributed algorithm in such a way that in ``most" of the rounds (this is an oversimplification, but provides the right intuition), each node reads the messages of only few neighbors. Ideally, this would have been reading the messages of only $\Theta(1)$ neighbors per round. But we do not achieve that. Instead, our adaptation is such that if in a given round a node needs to read the messages of many, say $2^{\Theta(k)}$ neighbors, this is needed only for messages that were sent some $\Theta(k)$ rounds earlier. Because of this, in a sense, we can amortize the high $2^{\Theta(k)}$ increase factor in the influence to the elapsed $\Theta(k)$ time. This allows us to show that overall, each node is influenced by only $\Delta \cdot 2^{\Theta(T)}$ other nodes. The LCA will be able to simulate the local algorithm basically by (adaptively) tracing which messages need to be read, and that brings the query complexity down to $\Delta \cdot 2^{\Theta(T)} = \poly(\Delta)$. 

We note that the best previously known query complexity $\Delta^{O(\log\log \Delta)} \poly(\log n)$\cite{ghaffari2019sparsifying} also was based on a (different) adaptation of the algorithm of \cite{ghaffari2016MIS}. Moreover, somewhat similar properties turned out to be useful for parallel algorithms~\cite{ghaffari2019sparsifying, ghaffari2021PRAMMIS}. The LCA of \cite{ghaffari2019sparsifying} simulated the local distributed algorithm by breaking the $\Theta(\log \Delta)$ rounds of the local algorithm into batches and ensuring that the simulation is done in such a way that the degree per batch drops to at most exponential in the batch length. For instance, in the basic version of the algorithm where there is only one scale, we have $\Theta(\sqrt{\log \Delta})$ batches each of length $\sqrt{\log \Delta}$ rounds, and the simulation ensures that the degree per batch drops to $2^{\Theta(\sqrt{\log \Delta})}$. However, this approach comes at a cost and the best achievable query complexity is $\Delta^{O(\log\log \Delta)}$. In particular, even in the most efficient version, which includes $\Theta(\log\log \Delta)$ scales of batches, where the $i^{th}$ scale batches are of length $2^{i}$, we lose a $\Delta$ factor in the query complexity per scale. That leads to the query complexity of $\Delta^{O(\log\log \Delta)}$. To achieve a $\poly(\Delta)$ algorithm, here, we directly design the distributed local algorithm in such a way that each node is influenced (through chains of read messages) by at most $\poly(\Delta)$ nodes. 

We are hopeful that techniques similar to our approach can be developed for other graph problems, by devising local distributed algorithms that read only few messages in most rounds and thus have much lower causal influence per node than the naive $\Delta^{T}$ bound. This would provide a (hopefully general) recipe for achieving (centralized) local computation algorithms with small query complexity, well below the Parnas-Ron direct simulation principle\cite{parnas2007approximating}.

Finally, we note that there is some superficial similarity between one numerical aspect of our analysis and that of a matching approximation of Kapralov et al.\cite{kapralov2020space} in the streaming model and the local computation algorithms model. Very roughly speaking, in both places, we effectively make the number of neighbors on which a recursion is invoked to be a constant (in an amortized sense). To the best of our understanding, the similarity stops at this level. In particular, the algorithm of Kapralov et al.\cite{kapralov2020space} provides no solution for the MIS problem, and not even for the much simpler special case of MIS on line graphs, i.e., the \emph{maximal} matching problem. The strict maximality requirement makes the MIS problem much harder technically.

\subsection{Preliminaries}\label{subsec:prelim}
\vspace{-0.3cm}
\paragraph{Distributed LOCAL model\cite{Linial1992LocalityID}} Consider an undirected graph $G=(V, E)$ with $n=|V|$ where there is one processor on each node of the graph. Each processor/node has a unique identifier in $\{1, \dots, n^{10}\}$. Initially, each node knows only its neighbors. Per round, each node can send one arbitrary-size message to all of its neighbors and receives their messages (we discuss which messages are read, soon). In the end, each node $v$ should know its own part of the output. For instance, in the MIS problem, node $v$ should know whether $v$ in the computed MIS or not.

\paragraph{Reading Messages Selectively, and Influence} In this paper, in order to devise a \local-model algorithm suitable for efficient LCA simulation, we take some extra care in how the communications of the \local model are performed. Per round, each node sends its message to all of its neighbors. However, on the receiving end, each node will choose which of the received messages it reads. In particular, each node will have a set of ``relevant neighbors" for each round (which was determined in earlier rounds) and it will read only the content of messages received from those relevant neighbors. So, the node discards the content of the messages from other neighbors and even ignores whether any of the other neighbors sent a message or not in that round. If a node $v$ has read a message of a neighbor $u$, we say that $u$ has \emph{influenced} $v$ (directly). If node $v$ reads a message of a neighbor $u$ who is influenced by a node $w$, we say node $v$ is influenced by node $w$. 

\paragraph{Notations} We use $\Gamma(v)$ to denote the set of neighbors of node $v$ in graph $G$, and $\Gamma^{r}(v)$ for $r\geq 1$ to denote the set of all nodes within distance at most $r$ from $v$. Notice that $\Gamma^{1}(v)=\Gamma(v)\cup \{v\}$. We also use $\Gamma(I)$ for a set $I\subset V$ of vertices to indicate all vertices that have a neighbor in $I$, i.e., $\Gamma(I) = \cup_{v\in I} \Gamma(v)$.  For a positive integer $T$, we use $[T]$ as a shorthand for $\{1, 2, \dots, T\}$.

\section{The Algorithm}
In \Cref{subsec:LOCALalgo}, we first describe a near-maximal independent set algorithm in the LOCAL model of distributed computing. Then, in \Cref{subsec:observations}, we discuss some observations and properties of this algorithm. Afterward, in \Cref{subsec:LCA}, we discuss how, thanks to these properties, the algorithm can be simulated efficiently in a sequential manner, in the setting of Local Computation Algorithms. Then, in \Cref{subsec:Analysis}, we present the analyses which prove this efficiency and that the algorithm indeed computes a near-maximal independent set. We conclude in \Cref{subsec:completeLCA} by discussing how this algorithm can be turned into a maximal independent set algorithm.
\subsection{The LOCAL Algorithm}
\label{subsec:LOCALalgo}
\paragraph{General Outline} As the core ingredient, we describe a randomized algorithm that computes a near-maximal independent set $I$. The near-maximality is in the sense that if we remove set $I$ and its neighbors $\Gamma(I)$ from the graph, with probability $1-1/\poly(n)$, in the subgraph induced by the remaining vertices $V\setminus (I \cup \Gamma(I))$, each connected component has at most $\poly(\Delta) \log n$ vertices. Moreover, each node $v$ is in this remaining part with probability at most $1/\poly(\Delta)$. Because of the former property, it will be easy to turn this near-maximal independent set algorithm into a maximal independent set algorithm, with only a factor of $\poly(\Delta) \log n$ increase in the number of queries, simply by identifying the remaining component of the questioned vertex and using the deterministic greedy MIS procedure on this small component. This is explained in \Cref{subsec:completeLCA}. In fact, for each vertex $v$ and any $\eps>0$, with probability $1-\eps$, the component of $v$ has size at most $\poly(\Delta)\log(1/\eps)$. Thus, for each vertex $v$, with probability $1-\eps$, the query complexity of the LCA when questioning vertex $v$ is at most $\poly(\Delta)\log(1/\eps)$. 

To provide the aforementioned near-maximal independent set algorithm, we present a $T$-round randomized algorithm in the LOCAL model of distributed computing that computes a near-maximal independent set with the same guarantees as mentioned above, where $T=\Theta(\log \Delta)$. The base of this algorithm is an adaptation of Ghaffari's $T$-round algorithm~\cite{ghaffari2016MIS}. The differences will be in how we adapt the algorithm so that it can be simulated with only $\poly(\Delta)$ query complexity in the LCA model, unlike the $\Delta^{\Theta(\log \Delta)}$ query complexity of the algorithm of~\cite{ghaffari2016MIS}. We will discuss those differences later, after reviewing the common and basic ingredients. 

The basic operation in both algorithms is a probabilistic marking process, which determines the nodes that attempt to join the independent set in one round. We next discuss this basic marking process and mention how we will determine all the randomness used for this process at the very beginning of the algorithm. We will then  discuss how our algorithm departs from that of ~\cite{ghaffari2016MIS} and leads to a $\poly(\Delta)$ query complexity in the LCA model.

\paragraph{Marking process} The basic operation in each round of the algorithm is a randomized marking process, which determines the nodes that attempt to join the independent set $I$. Concretely, per round $t$ each node $v$ is marked with a certain probability $p_{t}(v)$. Then, in that round, if node $v$ is marked and none of its neighbors is marked, node $v$ gets added to the independent set $I$, and we remove node $v$ and all of its neighbors from the graph. 

Each node $v$ starts with a marking probability $p_{1}(v)=1/2^{\lceil\log \Delta \rceil+1}$; this is the marking probability in the first round. During the algorithm, in each round $t\in [1, T]$, we will set either $p_{t+1}(v)\gets p_{t}(v)/2$ or $p_{t+1}(v)\gets \min\{2p_{t}(v), 1/2\}$. We will later discuss how this decision is done, in the main body of the algorithm. Let us now remark a smoothness property of this probability: the probability $p_{t}(v)$ changes by at most a $2$ factor per round, and concretely, for any two rounds $t, t'\in [T]$ where $t'\leq t-1$, we have $p_{t}(v) \leq p_{t'} \cdot 2^{(t-t')}$. 

Next, we explain our way of viewing the randomness used in this marking process, and how we fix it at the beginning of the algorithm.

\paragraph{Fixing the randomness of marking} At the beginning of the algorithm, for each node $v\in V$ and each round $t\in [1, T]$, we choose a uniformly random number $\rho_{t}(v)\in [0,1]$, independently of all other vertices and all other rounds. We note that a random number with $\Theta(\log \Delta)$ bits of precision suffices. The interpretation is this: suppose we want to mark node $v$ with a given probability $p_{t}(v)\in [0, 1]$ in round $t$. Then, we consider $v$ marked in round $t$ if and only if $\rho_t(v)\leq p_{t}(v)$. We fix all the randomness at the beginning of the algorithm. Furthermore, this is the only randomness that we will use in the algorithm. Intuitively, having fixed all the randomness at the beginning, we are dealing with a deterministic process. This determinism is helpful when we want to devise an LCA simulation of this LOCAL algorithm and discuss that the LCA performs the same computation and outputs the same independent set.

\medskip

\paragraph{Intuitive discussions---query complexity in LCA simulation of the LOCAL algorithm} 
As mentioned before, our near-maximal independent set algorithm is a $T$-round process in the \local model. We devise the algorithm in such a way that it can be simulated efficiently in the LCA model, with query complexity $\poly(\Delta) \ll \Delta^{T}$. We emphasize that this query complexity statement is deterministic. The probabilistic aspect will be only in the near-maximality guarantee. 
 
Concretely, we want that the behavior of each node depends on only $\poly(\Delta)$ other nodes; this allows us to simulate the algorithm in the LCA model with a $\poly(\Delta)$ query complexity. Recalling the definition of influence from \Cref{subsec:outline} and \Cref{subsec:prelim}, the new $T$-round algorithm is devised such that each node is influenced by only $\Delta \cdot 2^{\Theta(T)}$ other nodes. Informally, the $2^{\Theta(T)}$ factor accounts for an increase of a constant factor per round, which comes roughly from each node being influenced directly by constant many of its neighbors. The $\Delta$ factor is the influence in the very first round, as in that round each node checks all the messages received from all of the neighbors (e.g., to initialize the sets of ``relevant neighbors", which we will discuss soon).

In particular, unlike in the original algorithm\cite{ghaffari2016MIS} where each node reads the messages of all of its up to $\Delta$ neighbors in each round, we would like to make nodes read messages of only few neighbors per round. As mentioned above, this would be ideally just a constant number of neighbors per round, with the exception of the very first round where we allow reading messages of all the up to $\Delta$ neighbors. We do not achieve this exact property, but we get it in a certain amortized sense: The new distributed algorithm is devised such that if $v$ has to read many messages in round $t$, say $2^{\Theta(k)}$ many messages for some number $k$, then these should be messages that were sent much earlier, in rounds before $t-\Theta(k)$. Said differently and informally, the algorithm is devised such that, if there are many neighbors of node $v$, say $2^{\Theta(k)}$ that are attempting to join the independent set in a given round, then node $v$ can sleep for $\Theta(k)$ rounds without reading their messages. The concrete way through which we achieve this is by building a certain set $N_{t}(v)$ of \emph{relevant neighbors} for round $t$, which is initialized to include all nodes that might get marked in round $t$ and then refined over time as we get closer to round $t$, to smaller and smaller sets. If at some point $N_{t}(v)$ is ``too large" compared to the related size threshold, then node $v$ sleeps for a number of rounds including round $t$. 

Next, we describe the algorithm. Afterward, in \Cref{subsec:observations}, we present some observations regarding the size of $N_{t}(v)$ and the number of messages read in each given round sent in each earlier round. These observations formalize the above intuitive statements. We later use these properties in \Cref{subsec:LCA} to argue that we can develop an LCA simulation of the near-maximal indepdendent set algorithm of the \local model with query complexity $\poly(\Delta)$.

\medskip
\paragraph{The Algorithm---Initialization}
We next describe the LOCAL algorithm for a near-maximal independent set. We first discuss round 1 which involves a special initial preparation for all the future rounds $[1, T]$.  We comment that, in round $1$, each node $v$ sends and receives messages to and from all of its up to $\Delta$ neighbors, and reads all the received messages.

\smallskip
\begin{mdframed}
\textbf{Initialization (in round 1)--Determining Relevant Neighbors for Future Rounds.} In each round of the algorithm, each node $v$ will read the messages sent by only a selected subset of its neighbors. In the beginning of round $1$, node $v$ creates a \emph{set of relevant neighbors} $N_{t}(v)$ for each future round $t\in [1, T]$. In particular, in round $1$, we build the set $N_{t}(v)$ for each round $t\in [1, T]$ as follows: Each node $u$ sends $\rho_{t}(u)$ for each round $t\in [1, T]$ to each of its neighbors. Each node reads all received messages in this round. So, each node $v$ learns $\rho_{t}(u)$ for each neighbor $u\in \Gamma(v)$ and for each round $t\in [1, T]$. Then, node $v$ includes each neighbor $u\in \Gamma(v)$ in the set $N_{t}(v)$ if and only if $\rho_{t}(u) \leq p_{1}(u) \cdot 2^{(t-1)} \leq \frac{2^{t-1}}{2\Delta}$. Notice that since $p_{t}(u) \leq p_{1}(u) \cdot 2^{(t-1)}$, and given how we use the random numbers $\rho_{t}(u)$ for the marking decision, this preliminary set $N_{t}(v)$ includes all neighbors of $v$ that can be marked in round $t$.

By default, we would think that node $v$ is allowed to read messages from all neighbors in $N_{t}(v)$ during round $t$. However, these sets $N_{t}(v)$ might be too large and that would require node $v$ to read messages from many neighbors (e.g., to see whether any of them is marked). To circumvent that, we do as follows: for each round $t\in [1,T]$, if we have $|N_{t}(v)| > 2^{5(t-1)}+K$, then let $z$ be the greatest integer such that $|N_{t}(v)| > 2^{5(t+z-1)}+K$. In that case, we declare node $v$ \emph{sleeping} for rounds $[t, t+z]$. Here, $\delta\in(0, 0.01)$ is a desirably small positive constant and $K = 20\log( 1/\delta)$ is a large constant. 

\medskip
\noindent \textbf{Comment 1:} We see later how sleeping affects the algorithm. For now, consider this intuition: a node sleeping in a given round will not attempt to be added to the independent set. Therefore, it also does not need to read in that round whether its neighbors are marked (the node will read this information later, once it is not sleeping anymore, to see if any of those neighbors joined the independent set or not).

\smallskip
\noindent \textbf{Comment 2:} In the course of the algorithm, until we reach round $t$, we will further refine the set $N_{t}(v)$ in a gradual manner. This will be explained soon in the main body of the algorithm. Also, this will happen for all such sets $N_{t}(v)$ for different $t$ simultaneously. We emphasize that we will never add any new vertex to the set $N_{t}(v)$, but we may remove vertices from it.
\hspace*{5mm} 
\end{mdframed}
\bigskip

\paragraph{The Algorithm---Main Body} We next describe the main part of the algorithm, which happens during rounds 1 to $T$. Some intuition or remarks are provided in the footnotes.
\smallskip
\begin{mdframed}
\paragraph{Main body of the near-maximal independent set algorithm (rounds $1$ to $T$)} Let $\delta\in(0, 0.01)$ be a desirably small positive constant, and $K = 20\log_{2}(1/\delta)$; intuitively, we think of $K$ as a large enough constant. We describe the process in a round $t\in [1,T]$ for a node $v$, which involves four steps. All vertices perform the process of round $t$ simultaneously.

\begin{itemize}
    
    \item\textbf{Step 1:} Suppose that $v$ is not sleeping in round $t$. If node $v$ was sleeping for rounds $[t', t-1]$ and is not sleeping in round $t$, then node $v$ reads the messages sent by its neighbors in $N_{r}(v)$ for any $r\in [t', t-1]$. If any such neighbor $w\in N_{r}(v)$ had joined the independent set $I$, then $v$ is considered dead, it stops participating in the rest of the algorithm and sends a message to all of its neighbors describing that $v$ is dead. 
    
    \item\textbf{Step 2:} In round $t$, for each round $t''\in [t, T]$ such that $v$ has not been declared sleeping in round $t''$, we do as follows\footnote{We emphasize that this this step happens even if node $v$ is declared sleeping for the current round $t$.} (sequentially, going over different rounds $t''$ one by one): 
    
        \begin{itemize}
            \item For each earlier round $r$ from $t-2(t''-t)$ to $t-(t''-t)-1$, node $v$ does the following\footnote{The reason for the choice of this $[t-2(t''-t), t-(t''-t)-1]$ interval will become clearer when we discuss the size of the relevant neighbor set in \Cref{clm:NeighborsIfNOTASLEEP}; we use this to bound the number of messages that a node reads in \Cref{clm:messagesRead}, and later in the simulation of the algorithm in the LCA model, in \Cref{lem:LCA}. Intuitively, we need that this interval has length $\Theta(t''-t)$, so that the guarantee that we have about the size of $N_{t''}(v)$ from checks in earlier rounds (as $v$ is not declared sleeping for round $t''$) allows us to read all messages sent in these rounds while having a bound on $|N_{t''}(v)|$ that is $2^{O(t''-t)}$ + K. On the other hand, we could end the interval somewhat later, e.g., in round $t-1$, but we would not gain anything asymptotically, because even after reading messages sent in round $t-1$, the best bound that we would have on the size of $|N_{t''}(v)|$ would be $2^{ 5(t''-t)}$ + K.} (sequentially, going over different rounds $r$ one by one):
                \begin{itemize}
                \item Node $v$ reads the messages that were sent during round $r$ by the neighbors in the set $N_{t''}(v)$. We emphasize that this is done only if node $v$ has not been declared sleeping for this round $t''$, so we know that $N_{t''}(v)$ cannot be too large--- the exact bounds are discussed later. Then, node $v$ checks and refines $N_{t''}(v)$ by keeping in  $N_{t''}(v)$ only vertices $u$ such that $u$ is not dead and for which $\rho_{t''}(u) \leq p_{r}(u) \cdot 2^{(t''-r)}$. 
                
                In the end of processing all messages sent by neighbors in  $N_{t''}(v)$ during round $r$, if we still have $|N_{t''}(v)| > 2^{5(t''-r)} + K$, then we declare node $v$ sleeping in rounds $[t'', t''+(t''-r)]$ and we end this for loop, thus proceeding to the next value of $t''$. Otherwise, we proceed to the next value of $r$. 
            \end{itemize}
        \end{itemize}
        
    \item\textbf{Step 3:} Suppose that node $v$ is sleeping in round $t$. Then, we call it marked in this round iff $\rho_{t}(u) \leq p_{t}(v)$. Moreover, we set $p_{t+1}(v)\gets p_{t}(v)/2$, send $p_{t+1}(v)$ to the neighbors, and proceed to the next round, thus skipping step 4. 
    
    \item\textbf{Step 4:} Suppose that node $v$ is not sleeping in round $t$. Then, we consider node $v$ marked in round $t$ iff $\rho_{t}(v)\leq p_{t}(v)$. Node $v$ sends a message declaring whether it is marked or not to its neighbors. If node $v$ is marked and no neighbor in $N_{t}(v)$ is marked, then node $v$ joins the independent set $I$ and sends a message to all of its neighbors describing that $v$ joined the independent set $I$. Suppose $v$ did not join the independent set $I$ (either because it was not marked, or because it had a marked neighbor). Then, if any neighbor in $N_{t}(v)$ was marked, we set $p_{t+1}(v)\gets p_{t}(v)/2$, and otherwise we set $p_{t+1}(v)\gets \min\{2p_{t}(v), 1/2\}$. Node $v$ sends a message to its neighbors describing the value of $p_{t+1}(v)$.       
\end{itemize}
\end{mdframed}
\medskip
\subsection{Observations regarding the number of messages read per round}\label{subsec:observations}
The following two observations capture a crucial property of the algorithm, which is one of the key reasons for the perhaps strange design of the LOCAL algorithm. Informally, they show that the only time that a node needs to read many messages in a particular round $t$, say $2^{\Theta(k)}$ many messages, is when it is reading messages of a much earlier round $t-\Theta(k)$. Said differently, in a round $t$, we need to reach only constant messages from the previous round, or indeed the previous $O(1)$ rounds, and more generally at most $2^{\Theta(k)}$ messages from an earlier round $t-\Theta(k)$. Notice that this is quite different than usual LOCAL algorithms, where per round $t$, each node should read up to all the $\Delta$ messages sent by the neighbors even in round $t-1$ (and tracing those messages naively suggest that, in round $t$, the node is influenced by up to $\Delta^t$ nodes overall, directly or indirectly). One can see that, because of the new bound, the number of nodes that influence each particular node can be bounded by $\poly(\Delta) \cdot 2^{\Theta(T)}$. Instead of formally discussing this point, we later see how these observations allow us to obtain an LCA with $\poly(\Delta) \cdot 2^{\Theta(T)}$ query complexity for simulating the above $T$ round algorithm. 

\begin{claim}\label[claim]{clm:NeighborsIfNOTASLEEP}
At the end of round $t$, the set of relevant neighbors $N_{t''}(v)$ for a future round $t''> t+1$ has size at most $2^{10(t''-t+1)}+K$, unless $v$ has been declared sleeping for round $t''$. We emphasize that this statement is deterministic.
\end{claim}
\begin{proof}
In round $t$, node $v$ examines the messages sent by nodes in $N_{t''}(v)$ during rounds $[t-2(t''-t), t-(t''-t)-1]$, round by round, and refines $N_{t''}(v)$ accordingly. In particular, if at the end, in round $r=t-(t''-t)-1$, we have $|N_{t''}(v)| > 2^{5(t''-r)} + K = 2^{10(t''-t+1)} + K$, then node $v$ is declared sleeping for rounds $[t'', t''+(t''-r)]$.
\end{proof}

\begin{claim}\label[claim]{clm:messagesRead}
In any round $t\in [T]$, for any past round $t' \leq t-1$, any node $v$ reads the messages of at most $2^{25(t-t'+1)}+K(t-t')$ neighbors sent during round $t'$.
\end{claim}
\begin{proof}
In round $t$, node $v$ needs to read the message sent by a neighbor $u$ in round $t' \leq t-1$ for two purposes: (A) to perform step 1, which requires knowing whether $u$ joined the independent set $I$ in round $t'$ or not for $u\in N_{t'}(v)$, and (B) to perform step 2, which requires, for every future round $t''\geq t$, knowing the marking probability of node $u\in N_{t''}(v)$ in each round $r\in [t-2(t''-t), t-(t''-t)]$. We discuss these two cases separately. 
\begin{itemize}
    \item \textbf{Part (A), messages of step 1.} Suppose that $v$ is not sleeping in round $t$, as otherwise it reads no messages in step 1. Then, for the earlier round $t'\leq t-1$, during step 1 of round $t$, node $v$ will read messages sent in round $t'$ by nodes in $N_{t'}(v)$ to see if any of them joined the independent set $I$ during round $t'$.  Since the set $N_{t'}(v)$ does not grow over time, it suffices that node $v$ reads the message of each neighbor $u$ sent during round $t'$ only if $u\in N_{t'}(v)$ in the earlier round $q = t'- (t-t')$. Notice that $t'-q=t-t'$. We have two cases: (I) Suppose $q\geq 1$, which means $t'>t/2$. Since node $v$ was not declared sleeping for round $t$ by the end of round $q$, by \Cref{clm:NeighborsIfNOTASLEEP}, we get that in the of round $q$ we must have had $|N_{t'}(v)| \leq 2^{10(t'-q+1)} + K = 2^{10(t-t'+1)} + K$. (II) Suppose that $q\leq 0$, which means $t'\leq t/2$. Notice that in this case talking about an event that happened in round $q$ is meaningless as $q\leq 0$. But here we can rely on the initialization performed at the beginning of round $1$. In particular, if at that point we had $|N_{t'}(v)| > 2^{5(t-1)} + K$, then node $v$ would have been declared sleeping for rounds $[t', t]$ which includes round $t$. Since $v$ is not sleeping in round $t$, we must have had $|N_{t'}(v)| \leq 2^{5(t-1)} + K$. We have $t'\leq t/2$ which means $(t-t')\geq t/2.$ Hence, we can conclude that $|N_{t'}(v)| \leq 2^{10(t-t'+1)} + K.$

    \item \textbf{Part (B), messages of step 2.} During step 2 of round $t$, mode $v$ might read the messages sent in a past round $t'\leq t-1$, by neighbors in $N_{t''}(v)$ for different future rounds $t''\geq t$. First, let us discuss a fixed future round $t''\geq t$. Suppose that $t''$ was not declared a sleeping round for $v$ in the rounds before $t$, as otherwise we do not process any messages from $N_{t''}(v)$ in step 2 of round $t$. Let $q=t-(t''-t)$. If $q\leq 1$, then there are no messages from rounds $[t-2(t''-t), t-(t''-t)-1]$ to be checked in step 2 of round $t$ for future round $t''$, because $t-(t''-t)-1=q-1\leq 0$ and there is no such round. Suppose $q\geq 1$. Node $v$ was not declared sleeping in round $t''$ when we refined $N_{t''}(v)$ during step 2 of the earlier round $q=t-(t''-t)$. Hence, by \Cref{clm:NeighborsIfNOTASLEEP}, in the end of round $q$, we must have had $|N_{t''}(v)|\leq 2^{10(t''-q+1)}+K \leq 2^{20(t''-t+1)}+K$. Thus, also in the beginning of round $t$ (even before performing any of its steps), we must have $|N_{t''}(v)|\leq 2^{10(t''-q+1)}+K \leq 2^{20(t''-t+1)}+K$.
    
    In round $t$, we read messages from neighbors in $N_{t''}(v)$ sent during rounds $[t-2(t''-t), t-(t''-t)-1]$, and we gradually filter $N_{t''}(v)$ accordingly. In this lemma, we are interested in how many messages sent in the particular round $t'$ are read in round $t$. So, support $t'$ is a round in this interval $[t-2(t''-t), t-(t''-t)-1]$ for which we read during round $t$ the messages that were sent by nodes in $N_{t''}(v)$ during round $t'$. Since we have upper bounded  $|N_{t''}(v)|\leq 2^{10(t''-q+1)}+K \leq 2^{20(t''-t+1)}+K$ at the beginning of round $t$, the number of messages read during round $t$ sent during earlier round $t'$ by neighbors in $N_{t''}(v)$ is at most $2^{20(t''-t+1)}+K.$ 
    
    Now, we can sum the above over different choices of $t''$. Notice that in the above, $t''$ was a round in the interval $[\frac{(t-t')}{2}+t, (t-t') + t]$. Hence, the total number of messages that node $v$ will read during round $t$ sent at the earlier round $t' \leq t-1$, summed up over all the relevant choices of future rounds $t''\geq t$, is at most $$\sum_{t''=\frac{(t-t')}{2}+t}^{(t-t')+t} (2^{20(t''-t+1)} + K) \leq  2 \cdot 2^{20(t-t'+1)} + K\cdot (t-t')\leq 2^{25(t-t'+1)}+K(t-t').\qedhere$$
\end{itemize}
\end{proof}

\subsection{LCA Simulation}
\label{subsec:LCA}
We devise a (recursive) LCA that allows us to simulate the above LOCAL procedure, so that we can determine each node's status at the end of the algorithm using $\poly(\Delta)$ queries. We first provide an intuitive discussion and then present the formal result.

\paragraph{What needs to be done in a simulation} For any $t\in [1, T$], let us use $Q(t)$ to denote the subroutine (or oracle) that receives a node $v$ as an input and simulates its behavior for rounds $[1, t]$. Slightly abusing notation, we also use $|Q(t)|$ to refer to (the worst case over all vertices $v$ of) the number of queries during the run of $Q(t)$ called on node $v$. When $Q(t)$ is called upon node $v$, we need to perform three things: (1) Simulate node $v$ in rounds $[1, t-1]$, which can be done by calling $Q(t-1)$ on $v$. (2) For each round $r\in [1, t-1]$, simulate neighbors $N_{r}(v)$ and see if any of them joins the independent set $I$ in round $r$, which would mean we declare node $v$ dead by round $t$ (note that it is possible that a neighbor has joined the independent set $I$ in a much earlier round $r\ll t$ and node $v$ reads that message only in round $t$, because $v$ was sleeping before that). (3) We also need to simulate some of the neighbors of $v$ for rounds $[1,t-1]$; these are neighbors that may influence node $v$ in round $t$. For the intuitive discussion, let us focus on only this third part. This part itself can turn out to be too expensive in terms of query complexity, if not done right. 

\paragraph{Intuitive discussion of the naive simulation} For this third part, the question is which neighbors should be simulated for rounds $[1, t-1]$ and how many such neighbors are there? Let us simplify the problem for now and ignore that in round $t$ we also do some refinement work for future rounds $t''$. So, for now, we are just interested in the marking process and the independent set that is created up to (and including) round $t$, ignoring the preparation work for the future rounds.

Consider the procedure $Q(t)$. Before the beginning of round $1$, we have read the randomness of all the up to $\Delta$ neighbors of node $v$ and have formed the initial set $N_{t}(v)$ of neighbors that can influence node $v$ in round $t$. At that point, all that we know is that this set has size at most $2^{5t}+K$, as otherwise we would have put $v$ to sleep in round $t$ and thus it would not read the messages of its neighbors for the sake of the marking process of round $t$. Naively, we would think of calling $Q(t-1)$ on each of those neighbors. However, this would lead to the recursive complexity $|Q(t)| = \Delta + (2^{5t} + K) \cdot |Q(t-1)|> 2^{5t} \cdot |Q(t-1)|$. The solution to this recursion is $|Q(T)| \geq  2^{\Theta(T^2)}$. That is, to simulate the $T=\Theta(\log \Delta)$ round algorithm, this would have query complexity $2^{\Theta(\log^2 \Delta)} = \Delta^{\Theta(\log \Delta)}$, which is well beyond our target complexity of $\poly(\Delta)$, and essentially no better than a naive simulation of any $\Theta(\log \Delta)$ round algorithm. 

\paragraph{Remedy (intuitive but imprecise explanation)} Let us instead use the guarantees of the LOCAL algorithm, and in particular the observations described in \Cref{subsec:observations}. These allow us to gradually filter out $N_{t}(v)$. Suppose we simulate this initial set of $2^{5t} + K$ neighbors in $N_{t}(v)$ for $1$ round. At this point, the set $N_{t}(v)$ of relevant neighbors should shrink to $2^{5(t-1)} + K$ or otherwise we would declare node $v$ sleeping in round $t$ and thus would not need to simulate it. More generally, if we have simulated the previous set $N_{t}(v)$ of relevant neighbors for $t'$ many rounds, the remaining set size should drop to at most $2^{5(t-t')} + K$, as otherwise we would declare node $v$ sleeping in round $t$. Using this, we get the following recursion

$$|Q(t)| = \Delta + \bigg[(2^{5t} + K) \cdot |Q(1)| + (2^{5(t-1)} + K) \cdot |Q(2)| + \dots + (2^{10} + K) \cdot |Q(t-1)|\bigg].$$
Recall that $K$ is a large constant. The solution is $|Q(T)| \leq \Delta \cdot exp(\Theta( T))$, as can be proved by induction. Once we set  $T=\Theta(\log \Delta)$, this gives the overall query complexity of $\poly(\Delta)$ for (this partial) simulation the $T$-round LOCAL algorithm on any single node $v$. This discussion ignores some of the other computations that have to be performed for each node in a round, but provides the right intuition for the overall complexity. We next present the formal statement and analysis.

\begin{lemma}\label[lemma]{lem:LCA}
There is a local computation algorithm that simulates the above LOCAL algorithm. When questioned about any single node $v$, it performs $\poly(\Delta)$ queries to the graph, and $\poly(\Delta)$ computations, and answers whether $v \in I$ or $v\in \Gamma(I)$ or $v\notin (I \cup \Gamma(I))$. Here, $I$ is the independent set computed in the course of the LOCAL algorithm.
\end{lemma}
\begin{proof}[Proof of \Cref{lem:LCA}] 
For any $t\in [1, T$], let us use $Q(t)$ to denote the LCA subroutine that receives any node $v$ as an input and simulates its behavior for rounds $[1, t]$. Slightly abusing notation, we also use $|Q(t)|$ to refer to the number of queries during $Q(t)$ (worst case over all vertices $v$).

Furthermore, again with some abuse of notation, let us use $Q(0)$ to indicate the subroutine that includes all the initialization steps that a node performs before even the first round: setting its $p_{1}(v)$, forming the initial neighbor sets $N_{t}(v)$ for all $t\in [T]$ by checking the neighbors $u$ (concretely their random values $\rho_{t}(u)$) and declaring $v$ sleeping in rounds $[t, z]$ if $N_{t}(v) > 2^{5(t+z-1)}+K$. Notice that all of these can be performed using only $\Delta$ queries to the graph (to determine the neighbors). Hence, we will use $|Q(0)| = \Delta$ as our base case.

When $Q(t)$ is called upon node $v$, we first call $Q(t-1)$ on node $v$. This takes care of simulating node $v$ for rounds $1$ to $t-1$. What remains is to simulate node $v$ for round $t$, for all the four steps of the algorithm in that round. Of course, this will incur simulating some of the neighbors that may influence node $v$ in round $t$ for some of the previous rounds. We discuss this next, by examining each of the steps of the LOCAL algorithm separately.

\paragraph{Simulating Step 1} Let us first discuss step 1 of round $t$. If node $v$ is sleeping in round $t$ (indicated because of the simulation of $v$ for rounds $1$ to $t-1$), nothing needs to be done. Suppose this is not the case. Then, for each earlier round $r\in [1, t-1]$, we need to process $N_{r}(v)$ and determine whether any of those nodes joined the independent set $I$, which would thus make us declare node $v$ dead (i.e., it cannot join $I$). We divide this range into two parts of $[1, \lfloor t/2\rfloor]$ and $[\lfloor t/2\rfloor+1, t-1]$. 
\begin{itemize}
    \item \textbf{The First Half.} For each  $r\in [1, \lfloor t/2\rfloor]$, we do as follows: from the simulation of rounds up to $t-1$, which includes $Q(0)$, we know that $v$ has identified the initial set of $N_{r}(v)$ in round $1$. Notice that if $N_{r}(v)> 2^{5(t-1)}+K$, then node $v$ would have been declared sleeping for rounds $[r, t]$. Hence, we conclude that we must have $N_{r}(v) \leq 2^{5(t-1)} +K$. Then, we simply call $Q(r)$ on all of these at most $2^{5(t-1)} +K$ neighbors that are in $N_{r}(v)$ in round $1$. This costs $(2^{5(t-1)} + K)\cdot |Q(r)|$ queries for each round $r\in [1, \lfloor t/2\rfloor]$. 
    
    \item \textbf{The Second Half.} For each $r\in [\lfloor t/2\rfloor+1, t-1]$, we do as follows: from the simulation of rounds up to $t-1$, which includes $Q(0)$, we know that $v$ has identified the initial set of $N_{r}(v)$ in round $1$. Then, for each round $s$ from $1$ to $r-(t-r)$, we call $Q(s)$ on the set $N_{r}(v)$ of neighbors remaining after we have processed the neighbors in $N_{r}(v)$ up to round $s$. Notice that if $|N_{r}(v)| > 2^{5(r-s)} + K$, we would declare node $v$ sleeping in round $[r, r+(r-s)]$, which includes round $t$. The latter is because $s\leq r-(t-r)=2r-t$, which implies $r+(r-s) \geq 2r-(2r-t) = t$. Hence, we must have had $|N_{r}(v)| \leq 2^{5(r-s)} + K$ after having processed messages of rounds up to $s$. Hence, in total, for each round $r\in [\lfloor t/2\rfloor+1, t-1]$, the query complexity is at most $\sum_{s=1}^{r-(t-r)} (2^{5(r-s+1)} + K)\cdot Q(s).$  
    \end{itemize}
Summarized over all rounds of both halves, the query complexity to perform step 1 is at most

\begin{align} 
& &\bigg(\sum_{r=1}^{\lfloor t/2\rfloor} (2^{5(t-1)} + K) \cdot |Q(r)|\bigg) &+ \bigg(\sum_{r=\lfloor t/2\rfloor}^{t-1} \sum_{s=1}^{r-(t-r)} (2^{5(r-s+1)} + K) \cdot |Q(s)|\bigg) \nonumber \\
&\leq &\bigg(\sum_{r=1}^{\lfloor t/2\rfloor} (2^{10(t-r)} + K) \cdot |Q(r)|\bigg) &+ \bigg( \sum_{s=1}^{r-1} 2 (2^{5(t-s)} + K) \cdot |Q(s)|\bigg)\nonumber \\
&\leq &\bigg( \sum_{s=1}^{r-1} 3 (2^{10(t-s)} + K) \cdot |Q(s)|\bigg)\tag{1}
\end{align}

\paragraph{Simulating Step 2} Let us now discuss step 2 of round $t$. Here, the algorithm involves a separate process for each $t''\geq [t, T]$ that has not been declared sleeping for node $v$. Let us focus on one of these rounds. Let $q=t-(t''-t)$. If $q\leq 1$, then there are no messages from rounds $[t-2(t''-t), t-(t''-t)-1]$ to be checked in step 2, as there is no such round. This is also the case if $t''=t$, in which case the interval is also empty. Hence, suppose $q\geq 1$ and $t''\geq t+1$. Notice that at the end of round $q$, when we read messages from rounds up to $r=t-2(t''-t)-1$, node $v$ was not declared sleeping for round $t''$, which means we must have had $N_{t''}(v) \leq 2^{5(t''-r)} + K = 2^{5(3(t''-t)+1)} + K$. Next, for each earlier round $t'$ from $t-2(t''-t)$ to $t-(t''-t)-1$, the algorithm refines $N_{t''}(v)$ sequentially. If after processing messages of round $t'\in [t-2(t''-t), t-(t''-t)-1]$ we have $N_{t''}(v)> 2^{5(t''-t')} + K,$ then we declare node $v$ sleeping in rounds $[t'', t''+ (t''-t')]$ and thus we do not need to continue this process for that future round $t''$ (in this current round $t$). Notice that writing the relation in terms of $t''$, we can say $t''\in [\frac{(t-t')}{2}+t, (t-t') + t]$. Hence, the number of messages to be read after we have processed round $t'$ is at most $N_{t''}(v)\leq 2^{5(t''-t')}$. Because of this, we can upper bound the query complexity of processing the messages related to $N_{t''}(v)$ by 
\begin{align*}
   (2^{5(3(t''-t)+1)} + K)\cdot |Q(t-2(t''-t))| \;\; + \;\; \big(\sum_{t'=t-2(t''-t)+1}^{t-(t''-t)-1}  (2^{5(t''-t')} + K) \cdot |Q(t')|\big). \tag{$*$}
\end{align*}

Notice also that we need to calculate the summation of $(*)$ over different $t''$. Let us examine the summation of the above bound over different $t''$, by analyzing each of the two terms of $(*)$ separately. By setting $s=t-2(t''-t)$, the first term in ($*$) can be simply upper bounded as $(2^{20(t-s)} + K) \cdot Q(s)$. Hence, the summation over different $t''$ can be upper bounded by 
$\sum_{s=1}^{t-1} (2^{20(t-s)} +K) \cdot Q(s).$

For the second term in ($*$), for each particular $t'\leq t-2$, the summation over different $t''$ is upper bounded by
\begin{align*}
   \sum_{t''=\frac{(t-t')}{2}+t }^{(t-t') + t}  (2^{5(t''-t')} + K) \cdot |Q(t')| \leq (2 \cdot (2^{10(t'-t)}) + K (t-t'))\cdot |Q(t')|.
\end{align*}
Therefore, overall, the query complexity of step 2 of round $t$ can be upper bounded by 
\begin{align} 
\sum_{s=1}^{t-1} 2\cdot (2^{20(t-s)} +K(t-s)) \cdot |Q(s)|. \tag{2}
\end{align}

\paragraph{Simulating Step 3} This step happens only when node $v$ has been declared sleeping for round $t$ and in this case node $v$ is does not read any messages of the neighbors in this round, and thus there is no extra query needed.

\paragraph{Simulating Step 4} This step happens only when $v$ is not sleeping in round $t$. In that case, we determine whether $v$ is marked or not, i.e., whether it attempts to join the independent set $I$ in round $t$ or not. The only message that needs to be read from the neighbors is whether any neighbor in $N_{t}(v)$ also got marked or not. Since node $v$ was not declared sleeping for round $t$ in the earlier round $t-1$, we know that at that point we must have had $N_{t}(v)\leq 2^{10}+K$. Hence, the query complexity of this step is simply calling Q($t-1$) on each of these at most $2^{10}+K$ neighbors, which is upper bounded by $(2^{5}+K)\cdot |Q(t-1)|$.

Overall, we can thus upper bound the query complexity of simulating $v$ up to round $t$ as \begin{align*}
|Q(t)| \leq |Q(t-1)| +  \bigg( \sum_{s=1}^{t-1} 10 (2^{20(t-s)} + K (t-s)) \cdot |Q(s)|\bigg)  
\end{align*} 

Recall the base $Q(0)=\Delta$. Using a simple induction, one can show that the solution to the recursion is bounded as $Q(T) \leq \Delta \cdot 2^{CT}$, where $C$ is a sufficiently large constant. Hence, to simulate each node's behavior in the LOCAL algorithm that has round complexity $T=\Theta(\log \Delta)$, the query complexity is $\poly(\Delta)$.
\end{proof}

\subsection{Analysis of the LOCAL Algorithm}
\label{subsec:Analysis}
\paragraph{Intuitive Discussion} The general idea in the analysis is to show that despite the changes in the LOCAL algorithm (which are the whole reason that we can now achieve the $\poly(\Delta)$ query complexity), we can follow an analysis roughly along the lines of that of the original algorithm of \cite{ghaffari2016MIS}. Let us provide an intuitive outline. We comment this intuition is quite imprecise and inaccurate; the precise statements needs many definition and are presented after this intuition. We will show that, unless an unlikely event of probability at most $1/\poly(\Delta)$ happens, there will be a constant fraction of rounds that are good for $v$, in the sense that in each such round, either $v$ joins the MIS or a neighbor of $v$ joins the MIS. In particular, we will show that with probability $1-1/\poly(\Delta)$, all except a negligible minority of the decisions of a node to sleep are correct, meaning that the size of the neighborhood that was going to be marked in round $t$ was too large, and in the normal course of the algorithm the node would have not been able to join the independent set in this round. Having shown that, we use the dynamic of the algorithm to argue that in rounds where many neighbors of $v$ are attempting to join, either we are in a good round and there is a constant probability that at least one of them will succeed to join the independent set, or we are in such a round that the number of those attempting neighbor will shrink considerably. This will allow us to say that there cannot be a large fraction of rounds where too many of the neighbors of $v$ are attempting (and still none of them joins the independent set). Because of this, we will be able to show that unless there are a constant fraction of good rounds where neighbors of $v$ have a constant probability per round of joining $I$, there should be a constant fraction of good rounds such that $v$ has a considerable probability of attempting to join in this round and such that with constant probability no neighbor will block it. 

We next proceed to the analysis. For that, we first need to define a set of terminologies and notations. Afterward, we present a number of key claims about the behavior of the algorithm, and at the end, in \Cref{lem:localAnalysis} and \Cref{lem:shattering}, we show that each connected component of the graph induced by the remaining vertices $V\setminus (I\cup \Gamma(I))$ has size $\poly(\Delta) \log n$, with high probability.  

\subsubsection{Terminology and Round Classification}
Let us define $d_{t}(v) = \sum_{u\in \Gamma(v)} p_{t}(v)$. For each node $v$, we classify the rounds into a few different (not necessarily disjoint) possibilities: 
\begin{itemize}
     \item \textbf{Sleeping Mistake Rounds:} Recall that if in a round $t'$ and for a node $v$, we have $|N_{t}(v)|> 2^{5(t-t')} + K$ for some future round $t\geq t'+1$, then we declare node $v$ sleeping for rounds $[t, t+(t-t')]$. We say this was a mistake if at time $t'$ we had $\sum_{u\in \Gamma(v)} p_{t'}(u) \leq 2^{4(t-t')}.$ And in this case we have created $t-t' + 1$ rounds where node $v$ would be sleeping by mistake. If this mistake happens, we call the time interval $[t'+1, t+t-t']$ a \emph{sleeping mistake period} for node $v$, or simply a mistake period for $v$, and each round in this interval a \emph{mistaken} round for $v$.
     Any round that is not mistaken is called an \emph{intact} round.
     
     Similarly, before the beginning of round $1$, we examine $N_{t}(v)$ and if $|N_{t}(v)|> 2^{5(t-1)} + K$, we will put $v$ to sleep for some rounds. Let $z$ be the greatest integer such that  $|N_{t}(v)|> 2^{5(t+z-1)} + K$. Then, we put $v$ to sleep for rounds $[t, t+z]$. Noting that in round $1$ we must have had $\sum_{u\in \Gamma(v)} p_{1}(u) \leq \Delta \cdot \frac{1}{2\Delta} =\frac{1}{2}$ and as $p_{t}(u) \leq 2^{t-1} p_{1}( u)$, any such event for $z\geq 1$ is a mistake (i.e., significant deviation from the expectation of $|N_{t}(v)|$, and we consider the interval $[t, t+z]$ as a mistake period, and each of its rounds as a mistaken round for $v$.
    \item \textbf{Sleeping Round:} We say node $v$ is sleeping in round $t$ if $v$ has been put to sleep in round $t$ in one of the earlier rounds $t'\leq t-1$ (either correctly or by mistake).
    \item \textbf{Heavy Round:} We say node $v$ is heavy in round $t$ if $v$ is not sleeping in this round and when we reach round $t$, we have $d_t(v) =\sum_{u\in \Gamma(v)} p_{t}(v) \geq C_\delta=\log (1/\delta)$. 
    \item \textbf{Light Round:} We say node $v$ is light in round $t$ if $v$ is not sleeping in this round and when we reach round $t$, we have $d_t(v) =\sum_{u\in \Gamma(v)} p_{t}(v) \leq \delta$. 
    \item \textbf{Moderate Round:} We say node $v$ is moderate in round $t$ if $v$ is not sleeping in this round and when we reach round $t$, we have $d_t(v) \in (\delta, C_\delta).$
    \item \textbf{Wrong Down Move:} If in an intact round $t$, node $v$ is light but sets $p_{t+1}(v)  \gets p_{t}(v)/2$, we call this a \emph{wrong down move}.
    \item \textbf{Wrong Up Move:} If in an intact round $t$ where $d_{t}(v)\geq \delta$ and $t$ is not a good type-$2$ round for $v$, we have $d_{t+1}(v) > 0.7 d_{t}(v)$, we call this a \emph{wrong up move} for $v$.
        \item \textbf{Type-1 Good Round:} We call a round $t$ a good round of type-$1$ if this is a light or moderate round for $v$, but not a sleeping round, and we have $p_{t}(v)=1/2$.
    \item \textbf{Type-2 Good Round:} Let $H_{t}(v)\subseteq \Gamma(v)$ be the neighbors of $v$ that are heavy or sleeping in round $v$. Then, we call round $t$ a good round of type-$2$ for $v$ if $d_{t}(v)\geq \delta$ and we have $\sum_{u\in H_{t}(v)} p_{t}(v) \leq \frac{19}{20} d_t(v)$. Notice that node $v$ might be sleeping in such a round. 
    
\end{itemize}

\subsubsection{Key Lemmas}

\begin{claim}\label[claim]{clm:sleepMistakes} The probability of each round $t$ being a mistaken round for $v$ is at most $\delta^5$. Moreover, with probability $1-1/\poly(\Delta)$, the total number of mistaken rounds (which includes rounds in which node $v$ is sleeping by mistake) is at most $\delta T$ for a small constant $\delta\in (0, 0.01)$. Furthermore, this holds independent of the randomness of nodes $V\setminus \Gamma(v)$.
\end{claim}
\begin{proof} We first discuss the probability that each round $t$ is a sleeping mistake and upper bound this by $\delta^5$. Then, we argue that the total number of mistakes is at most $\delta T$, with probability $1-1/\poly(\Delta)$. We note that the latter does not follow immediately from the former as there are some dependencies.

\paragraph{Sleeping Mistake Probability Per Round} We first consider sleeping decisions made before the beginning of round $1$. Choose a round $t'\leq t$ and let us first consider the creation of $N_{t'}(v)$ before the beginning of round $1$, and the sleeping rounds it may cause, if $N_{t'}(v)$ is too large. Recall for that for each neighbor $u \in \Gamma(v)$, we have $p_{1}(u) \leq \frac{1}{2\Delta}$. Moreover, we include in $N_{t'}(v)$ neighbor $u$ if and only if we have $\rho_{t'}(u) \leq 2^{t'-1} p_{1}(u)$. Hence, $\mathbb{E}[|N_{t'}(v)|]\leq 2^{t'-2}$. We put $v$ to sleep for round $t'$, and some rounds after, only if we have $|N_{t'}(v)|> 2^{5(t'-1)} + K$, where $K=20\log(1/\delta)$. If $z$ is the greatest integer such that $|N_{t'}(v)|> 2^{5(t'+z-1)} + K$, then we put $v$ to sleep for rounds $[t', t'+z]$. Moreover, whether a neighbor $u$ is in $N_{t'}(v)$ or not simply depends on its own randomness $\rho_{t'}(u)$ and is independent of all other neighbors. Hence, by Chernoff bound, the probability of this event is at most $2\cdot exp(-(2^{5(t'+z-1)} + 20\log(1/\delta))/3) \leq \delta^5 exp(-\Theta(2^{5(t'+z-1)})) \leq \delta^{5} exp(-(2^{4(t+z-1)}))$. This means the probability that we put $v$ to sleep for round $t$ because of a particular prior round $t'\leq t$ is at most $\delta^5 exp(-(2^{4(t-1)}))$. Hence, by a union bound over all $t'\in[1, t]$, we get that the probability that we put $v$ to sleep for round $t'$ is at most $\delta^5 \cdot t\cdot exp(-(2^{4(t-1)})) \leq \delta^{5}$.

Next, let us consider sleeping decisions made in the course of the algorithm, in some round $t'\geq 1$. Consider a future round $t\geq t'+1$. We declare $v$ sleeping for rounds $[t, t+(t-t')]$ if in round $t'$ we had $|N_{t}|> 2^{5(t-t')}+K$. Notice that in this round $\mathbb{E}[|N_{t}|]  \leq \sum_{u\in \Gamma(v)} p_{t'}(v) \cdot 2^{(t-t')}$. Moreover, whether a neighbor $u$ is in $N_{t}(v)$ or not simply depends on its own randomness $\rho_{t}(u)$ and is independent of all other neighbors. Hence, by Chernoff, the probability that we actually had $\sum_{u\in \Gamma(v)} p_{t'}(v) \leq 2^{3(t-t')}$ and made this mistake of putting $v$ into sleep for rounds $[t, t+(t-t')]$ is at most $2\cdot exp(-(2^{5(t-t')}+K)/3) \leq \delta^{5} 2^{-(4(t-t'))}$, given that $K=20\log 1/\delta.$ Hence, by a union bound over all $t'\in [1, t-1]$, we get that the probability that round $t$ was a mistaken round for node $v$ is at most $\sum_{t'=1}^{t-1} \delta^{5} 2^{-(4(t-t'))} \leq \delta^5$.

\paragraph{The Total Number of Sleeping Mistakes}
Let us consider a round $t'$ and the set $N_{t'}(v)$ that we build for it and gradually refine it. Notice that, as discussed above, if $N_{t'}(v)$ crosses a certain threshold, then we put node $v$ into sleep for round $t'$ as well as a number of rounds after it. Notice also that this can happen only once; that is, only one time because of $N_{t'}(v)$ we put node $v$ into sleep for some rounds that includes round $t'$, and from that point on we do not refine $N_{t'}(v)$.  Let $R^{t'}(v)$ be the number of rounds that we put $v$ into sleep because of $|N_{t'}(v)|$ crossing a threshold in any previous round $t''\leq t'$. As discussed above, before the beginning of round $1$, we put $v$ into sleep for $z+1$ rounds --- that is, the period $[t', t'+z]$--- with probability at most $\delta^{5} exp(-(2^{4(t+z-1)}))$. If that happened, $v$ is sleeping in round $t'$ and we do not create any additional sleeping because of $|N_{t'}(v)|$. Otherwise, in each later round $r\in[1, t']$, we put $v$ into sleep for $t'-r+1$ rounds ---that is, the period $[t', t'+(t'-r)]$---with probability at most $\delta^{5} 2^{-(4(t'-r))}$. 

We can conclude that $\mathbb{E}[R^{t'}(v)] \leq \delta/2$.  Therefore, $\mathbb{E}[\sum_{t'=1}^{T} R^{t'}(v)] \leq \delta T/2$. 
Furthermore, the distribution of the random variable $R^{t'}(v)$ is stochastically dominated by a geometric distribution. More precisely, for any $k\geq 1$, we have $\Pr[R^{t'}(v)=k] \leq \frac{\delta}{2} \cdot 2^{-(2k)}$. Since the random variables $R^{t'}(v)$ are all independent for different values of $t'$, and as each is from a distribution stochastically dominating by a geometric distribution, using standard concentration arguments~\cite[Theorem 1.10.32]{doerr2018probabilistic}, we can conclude that $\sum_{t'=1}^{T} R^{t'}(v)$ is concentrated around its expectation and we thus have $\Pr[\sum_{t'=1}^{T} R^{t'}(v) \geq \delta T] \leq exp(-\Theta(\delta T)) = 1/\poly(\Delta)$.

Finally, notice that in this argument, we only examined the random variables that impact whether a neighbor $u\in \Gamma(v)$ is in $N_{t}(v)$ or not. Hence, the guarantee holds independent of the randomness of nodes $V\setminus \Gamma(v)$.
\end{proof}

We will sometimes make use of the following elementary helper lemma (very much like a union bound) to say that conditioning on an event that happens with probability almost $1$ does not change the probabilities too much. In particular, the latter for us will be the event that a particular round is intact, which we know happens with probability at least $1-\delta^5$, from \Cref{clm:sleepMistakes}.
\begin{lemma}\label[lemma]{lem:conditioning}
Consider an event $\mathcal{E}'$ which happens in round $t$ with probability $\Pr[\mathcal{E'}]$. Also, let $\mathcal{E}_{I}$ be another event, in our particular case, the event that round $t$ is intact for node $v$. Then, we have $\Pr[\mathcal{E}'|\mathcal{E}_{I}] \geq \Pr[\mathcal{E}'] - (1-\Pr[\mathcal{E}_{I}]).$
\end{lemma}
\begin{proof}
We have $\Pr[\mathcal{E}'|\mathcal{E}_{I}] \geq \Pr[\mathcal{E}' \cap \mathcal{E}_I] \geq \Pr[\mathcal{E}']- \Pr[\mathcal{E}' \cap \bar{\mathcal{E}_I}] \geq \Pr[\mathcal{E}'] - (1- \Pr[\mathcal{E}_{I}]) $.
\end{proof}

\begin{claim}\label[claim]{clm:wrongDown}
In each intact round $t$ where $v$ is light, with probability at least $1-2\delta$, we have $p_{t+1}(v) \gets \min\{2 p_{t}(v), 1/2\}$. In other words, the probability of each round being a wrong down move is at most $2\delta$. This statement holds independent of the randomness of nodes $V\setminus \Gamma(v)$.
\end{claim}
\begin{proof}
In an intact round $t$ where $v$ is light, we have $d_{t}(v)\leq \delta$. Hence, the probability that any of the neighbors of $v$ is marked is at most $\sum_{u\in \Gamma(v)} p_{t}(u) = d_{t}(u) \leq \delta$. Therefore, with probability at least $1-\delta$, no neighbor is marked. This proves the claim, because if node neighbor gets marked, we set $p_{t+1}(v) \gets \min\{2 p_{t}(v), 1/2\}$. Finally, we note that we should condition on that round $t$ is an intact round and $v$ was not mistakenly put to sleep for round $t$. Each round is an intact round with probability at least $1-\delta^5$, as shown in \Cref{clm:sleepMistakes}. Hence, by \Cref{lem:conditioning}, the probability that no neighbor is marked remains at least $1-\delta-\delta^5\geq 1-2\delta$.
\end{proof}

\begin{claim}\label[claim]{clm: heavy}
In each intact round $t$ where $v$ is heavy, with probability at least $1-\delta$, we have $p_{t+1}(v) \gets p_{t}(v)/2$. This statement holds independent of the randomness of nodes $V\setminus \Gamma(v)$.
\end{claim}
\begin{proof}
In an intact round $t$ where $v$ is heavy, we have $d_{t}(v)\geq C_{\delta}=\log_{2}({1/\delta})$. Notice that each neighbor $u\in \Gamma(v)$ is marked with probability $p_{t} (v)$, and regardless of whether $u$ is sleeping in this round or not. Thus, the probability that no neighbor of $v$ is marked is at most $\prod_{u\in \Gamma(v)} (1-p_{t}(u)) \leq e^{-\sum_{u\in \Gamma(v)} p_{t}(u)} = e^{-d_{t}(v)} \leq \delta/2$. Hence, with probability at least $1-\delta/2$, at least one neighbor is marked. Finally, we note that we should condition on that round $t$ is an intact round and $v$ was not mistakenly put to sleep for round $t$. Each round is an intact round with probability at least $1-\delta^5$, as shown in \Cref{clm:sleepMistakes}. Hence, by \Cref{lem:conditioning}, even conditioning on round $t$ being an intact round for $v$, we know that the probability that some neighbor of $v$ is marked is at least $1-\delta/2-\delta^5\geq 1-\delta$. Thus, with probability at least $1-\delta$, some neighbor of $v$ is marked and node $v$ sets $p_{t+1}(v) \gets p_{t}(v)/2$.
\end{proof}

\begin{claim}
In each intact round $t$ where $v$ is sleeping, we have $p_{t+1}(v) \gets p_{t}(v)/2$.
\end{claim}
\begin{proof}
Follows directly and deterministically, because the algorithm sets $p_{t+1}(v) \gets p_{t}(v)/2$ in step 3 of each sleeping round.
\end{proof}

\begin{claim}\label[claim]{clm:wrongUp}
In each intact round $t$ where $d_{t}(v)\geq \delta$ but $t$ is not a good type-$2$ round for $v$, with probability at least $1-10\delta$, we have $d_{t+1}(v) \leq 0.7 d_{t}(v)$. In other words, the probability of each round being a wrong up move is at most $10\delta$. This statement holds independent of the randomness of nodes $V\setminus \Gamma^2(v)$.

Moreover, a similar statement holds for sleeping rounds that are not good type-$2$ rounds, in an amortized sense: consider a node $v$ that is put to sleep in a round $t'$ for the period $[t, t+(t-t')]$, and this was not a mistake. We have two cases: 
\begin{itemize}
    \item Either each round $r \in [t, t+(t-t')]$ has $d_{r}(v)\geq \delta$, in which case either $r$ is a good type-$2$ round or, with probability at least $1-10\delta$, we have $d_{r+1}(v) \leq 0.7 d_{r}(v)$. Here, the statement depends only on the randomness of round $r$.
    \item Otherwise---if there is at least one round $r$ such that $d_{r}(v)\leq \delta$---we have $d_{t''}(v) \leq (0.7)^{t''-t'} d_{t'}(v)$ where $t''=t+(t-t')$.
\end{itemize}
\end{claim}

\begin{proof}
Consider a round $t$ in which $d_{t}(v)\geq \delta$. Suppose that this is not a good type-$2$ round for $v$, which means $\sum_{u\in H(v)} p_{t}(v) \leq \frac{19}{20} d_t(v)$ where $H(v)\subseteq \Gamma(v)$ is the neighbors of $v$ that are heavy or sleeping in round $t$. Let us call each neighbor $w \in H(v)$ who sets $p_{t+1}(w) \gets \{2p_{t}(w), 1/2\}$ wrong. By \Cref{clm: heavy}, any neighbor $w \in H(v)$ will set  $p_{t+1}(w) \gets p_{t}(w)/2$ with probability at least $1-\delta$. Thus, each neighbor in $w \in H(v)$ is wrong with probability at most $\delta$. Let $W\subset H(v)$ be those that were wrong.
Hence, we have $$\mathbb{E}[\sum_{w\in W(v)} p_{t+1}(w)] \leq  2\delta \cdot \sum_{w\in H(v)} p_{t}(w).$$ Therefore, by Markov's inequality, with probability at least $1-10\delta$, we have  $\sum_{w\in W(v)} p_{t+1}(w) \leq  \frac{1}{10} \cdot \sum_{w\in H(v)} p_{t}(w)$. Every neighbor $w'\in H(v)\setminus W(v)$ sets $p_{t+1}(w') = p_{t}(w')/2$. Hence, with probability at least  $1-10\delta$, we have 
\begin{align*}
d_{t+1}(v) &= \sum_{w\in \Gamma(v)} p_{t+1}(w) \\
&= \sum_{w\in \Gamma(v)\setminus H(v)} p_{t+1}(w)  + \sum_{w\in W(v)} p_{t+1}(w)  + \sum_{w\in (H(v)\setminus W(v))} p_{t+1}(w) \\
&\leq 2 \sum_{w\in \Gamma(v)\setminus H(v)} p_{t}(w) + \frac{1}{10} \cdot \sum_{w\in H(v)} p_{t}(w) + \sum_{w\in (H(v)\setminus W(v))} p_{t}(w)/2 \\
&\leq \frac{2}{20} \cdot \sum_{w\in H(v)} p_{t}(w) + \frac{1}{10} \cdot \sum_{w\in H(v)} p_{t}(w) + \frac{1}{2} \cdot \sum_{w\in H(v)} p_{t}(w) \\
&= 0.7 d_{t}(v) 
\end{align*}
Let us now consider the second part of the claim, where a node $v$ is put to sleep in a round $t'$ for the period $[t, t+(t-t')]$, and this was not by a mistake. The latter part means that at time $t'$, we had $d_{t'}(v)=\sum_{u\in \Gamma(v)} p_{t'}(v) \geq 2^{3(t-t')}$. We have already analyzed rounds $r$ in this period in which we have $d_{r}(v) \geq \delta$. In the complementary case, suppose that there is at least one round $r$ in the interval during which we have $d_{r}(v) \leq \delta$. Then at the end of the interval, we have $d_{r}(v) \leq \delta 2^{(t-t')} \leq (0.5)^{2(t-t')} d_{t'}(v) = (0.5)^{t''-t'} d_{t'}(v)$.
\end{proof}

\begin{claim}\label[claim]{clm:totalWrong} With probability $1-1/\poly(\Delta)$, the total number of wrong moves for a node $v$ is at most $20\delta T$. Moreover, this holds independent of the randomness of nodes $V\setminus \Gamma^2(v)$.
\end{claim}
\begin{proof}
By \Cref{clm:wrongDown,clm:wrongUp}, the probability of each intact round being a wrong down move or a wrong up move is at most $10\delta$. Furthermore, whether an intact round is a wrong up or down move depends only on the randomness used in the marking of that round, and is independent of other rounds. Hence, by Chernoff bound, With probability $1-1/\poly(\Delta)$, the total number of wrong moves for a node $v$ is at most $20\delta T$. In determining whether a round is a wrong up or down move, we only examined the randomness of $v$, its neighbors and the neighbors of its neighbors. Hence, the statement holds independent of the randomness of nodes $V\setminus \Gamma^2(v)$. 
\end{proof}

\begin{claim} \label[claim]{clm:good1}
In each intact good round of type-$1$, with probability at least $\delta^2/3$, node $v$ joins the independent set $I$. This statement holds independent of the randomness of nodes $V\setminus \Gamma(v)$.
\end{claim}
\begin{proof}
In an intact good round of type-$1$, we have $d_{t}(v)\geq C_{\delta}=\log(1/\delta)$ and $p_{t}(v)=1/2$. In this round, node $v$ is marked with probability $1/2$ and independent of that, the probability that none of its neighbors is marked is at least $\prod_{u\in \Gamma(v)}(1 - p_{t}(u)) \geq 4^{-\sum_{u\in \Gamma(v)} p_{t}(u)} = \delta^2$. Finally, we had conditioned on $t$ being an intact round. Since any round is intact with probability at least $1-\delta^5$,as shown in \Cref{clm:sleepMistakes}, we conclude that the probability that node $v$ joins the independent set $I$ is at least $\delta^2/2 - \delta^5 \geq \delta^2/3.$
\end{proof}

\begin{claim}\label[claim]{clm:good2}
In each intact good round of type-$2$, with probability at least $\delta^3/50$, one of the neighbors of node $v$ joins the independent set $I$. This statement holds independent of the randomness of nodes $V\setminus \Gamma^2(v)$.
\end{claim}
\begin{proof}
Let $H(v)\subseteq \Gamma(v)$ be the neighbors of $v$ that are heavy in round $t$ or sleeping. In an intact round of type-$2$, we have $d_{t}(v)\geq \delta$ and $\sum_{u\in H(v)} p_{t}(v) \leq \frac{19}{20} d_t(v)$. Thus, $\sum_{u\in (\Gamma(v)  \setminus H(v))} p_{t}(v) \geq \delta/20$. We scan the neighbors in  $(\Gamma(v)  \setminus H(v))$ one by one until we find the first marked neighbor. The probability that we find no marked neighbor is at most $\prod_{u\in \Gamma(v) \setminus H(v)} (1-p_{t}(u)) \leq e^{-\sum_{u\in (\Gamma(v)  \setminus H(v))} p_{t}(v)} \leq e^{-\delta/20} \geq 1-\delta/40$. In other words, with probability at least $\delta/40$ we find a marked neighbor in $(\Gamma(v)  \setminus H(v))$. Let $w$ be such a neighbor. Notice that $w$ is not sleeping and we have $d_{t}(w)\leq C_\delta = \log(1/\delta)$. Hence, the probability that none of the neighbors of $w$ marked is at least $\prod_{u\in \Gamma(w)}(1 - p_{t}(u)) \geq 4^{-\sum_{u\in \Gamma(w)} p_{t}(u)} = \delta^2.$ Therefore, with probability at least $\delta^3/40$, at least one of the neighbors of $v$ joins the independent set $I$. Finally, we had conditioned on $t$ being an intact round. Since any round is intact with probability at least $1-\delta^5$, as shown in \Cref{clm:sleepMistakes}, we conclude that the probability that a neighbor of node $v$ joins the independent set $I$ is at least $\delta^3/40 - \delta^5 \geq \delta^{3}/50.$
\end{proof}
\subsubsection{Proof of the Main Statement}
\begin{lemma}\label[lemma]{lem:localAnalysis}
In the $T=\Theta(\log \Delta)$ round algorithm, for each node $v$, with probability $1-1/\poly(\Delta)$, we have $v\in (I \cup \Gamma(I)).$ Moreover, this statement holds independent of the randomness of nodes $V\setminus \Gamma^2(v)$.
\end{lemma}
\begin{proof}
To prove the lemma, we show that with probability at least $1-1/\poly(\Delta)$, we have at least $0.1T$ good rounds for $v$ (in fact, without examining the marking outcome of the good rounds). Then, we use \Cref{clm:good1} and \Cref{clm:good2} to argue that with probability at least $1-1/\poly(\Delta)$, we must have $v\in (I \cup \Gamma(I)).$

First, note that by \Cref{clm:totalWrong} and \Cref{clm:sleepMistakes}, we have at most $25\delta T$ rounds that are wrong moves or mistakenly sleeping. Recall also that $\delta$ is a constant that we adjust to be desirably small.

Next, consider the set of all moderate, heavy, and sleeping rounds for $v$. Each such round is either a good type-$2$ round or a round $t$ in which we have $d_{t+1}\leq 0.7 \cdot d_{t}(v)$, or else it is either a wrong up move round or a mistakenly sleeping round (we separated the latter two kind as they are a tunably-small fraction, by decreasing $\delta$). Suppose we have less than $0.1T$ good rounds of type $2$, as otherwise we are done. In round $1$, we have $d_{t}(v)\leq 1/2$. In any round (and as particularly relevant in this argument, any good type-$2$ round or any wrong or mistaken round), we know that $d_{t+1}(v) \leq 2d_{t}(v)$. We have at most $0.1T+25\delta T$ such $2$-factor increases during moderate, heavy, and sleeping rounds. Any remaining round among moderate, heavy, and sleeping rounds is a decrease by an $0.7$ factor, i.e., where we have $d_{t+1}\leq 0.7 \cdot d_{t}(v).$ Since any increase round increases $d(v)$ by at most a $2$ factor and any decrease round reduces $d(v)$ by a factor of at most $0.7$, and given that $(0.7)^2 \leq 1/2$, we can have at most $0.3T + 75\delta T$ rounds in which $v$ is moderate, heavy, or sleeping.  

What remains are light rounds for $v$, and a tunably-small minority of wrong or mistaken rounds (namely at most $25\delta T$). So, we must have at least $0.7T-100\delta T$ light rounds. Each light round is either a wrong down move or we have $p_{t+1}(v)=\min\{2\cdot p_{t}(v), 1/2\}$. We start with $p_{1}(v)=\frac{1}{2\Delta}$ and in each light round that is not a wrong down move, we have $p_{t+1}(v)=\min\{2\cdot p_{t}(v), 1/2\}$. In each other round, we have $p_{t+1}(v)\geq p_{t}/2$. Hence, in at least $(0.7T-100\delta T) - (\log \Delta +1 + 0.3T + 75\delta T) \geq 0.1T$ light rounds, node $v$ is not sleeping and we have $p_{t}(v)=1/2$. Here, we have used that $\delta < 0.001$. These are good rounds of type $2$. 

In conclusion, node $v$ has at least $0.1T$ good rounds, with probability $1-1/\poly(\Delta)$. By \Cref{clm:good1} and \Cref{clm:good2}, in each good round, node $v$ joins the independent set $I$ or has a neighbor join $I$ with probability at least $\delta^3/50$. Hence, the probability that $v\notin (I \cup \Gamma(I))$ is at most $(1-\delta^3/50)^{0.1 T}\leq 1/\poly(\Delta)$, where we use that the constant in the length $T=\Theta(\log \Delta)$ is chosen large enough (e.g., greater than $1000/\delta^{3}$, where $\delta$ itself was a small constant fixed earlier).
\end{proof}
\begin{lemma}\label[lemma]{lem:shattering}
Let $I$ be the independent set computed in \Cref{lem:localAnalysis}. Let $H$ be the subgraph of $G$ induced by the subset $V\setminus (I \cup \Gamma(I))$ of vertices. Then, with probability $1-1/\poly(n)$, each connected component of $H$ has at most $O(\Delta^4 \log n)$ vertices.  Furthermore, for each node $v$ and for any $\eps>0$, the connected component of $v$ in $H$ has size at most $O(\Delta^4 \log (1/\eps))$.
\end{lemma}
\begin{proof}[Proof Sketch] This is a standard proof, which is sometimes called a shattering proof, and the first version of it appeared in Beck's breakthrough algorithm for Lovasz Local Lemma~\cite{beck1991algorithmic}. See, e.g., \cite[Lemma 4.6]{rubinfeld2011LCA} or \cite[lemma 4.2]{ghaffari2016MIS} for applications of this in the case of the MIS problem. To be self-contained, we provide a short and informal (and slightly imprecise) sketch. 

By \Cref{lem:localAnalysis}, each node $v$ is in $H$ with probability at most $1/\Delta^{20}$. We note that any probability $1/\Delta^c$, for any constant $c$, would be achievable in \Cref{lem:localAnalysis}, by adjusting the constant in the round complexity $T=\Theta(\log \Delta)$ of the LOCAL algorithm. The bound of $1/\Delta^{20}$ will suffice for our argument in this lemma. Also, for any two nodes that are $5$ or more hops apart, this statement holds independently, as for each of them it depends only on the randomness of nodes within its distance $2$. Let $T$ be a tree embedded in $G^10$ such that any two neighbors in $T$ have distance at most $10$ in $G$ and any two nodes in $T$ have distance at least $5$ in $G$. The probability that all vertices of $T$ are in $H$ is at most $(\frac{1}{\Delta^{20}})^{|T|}$. On the other hand, there are at most $n (4\Delta^{10})^{|T|}$ many trees of size $|T|$ that can be embedded in graph $G^{10}$. Namely, there are $n$ places as the root, $4^{|T|}$ different tree topologies, and at most $\Delta^{10}$ ways to choose each node in the tree once its parent has been chosen, among the less than $\Delta^{10}$ nodes within $10$ hops of its parent. Therefore, for $|T|=\Theta(\log n)$, with probability at least $1-1/\poly(n)$, there is no such tree in $H$. Hence, with probability at least $1-1/\poly(n)$, there is no connected component of size exceeding $\Theta(\Delta^4 \log n)$ in $H$, because any such component includes all the vertices of a tree of size $\Theta(\log n)$ embedded in $G^{10}$ with the aforementioned properties (as can be seen by a greedy argument of taking any vertex in the component of $H$ and removing others within its distance $4$, which are at most $\Theta(\Delta^4)$ many). 

Also note that if we fix node $v$ itself and we are only interested in the component of $v$, we have already fixed the root of the tentative embedding of $|T|$. Then, by setting $|T|=\Theta(\log 1/\eps)$, we get that with probability $1-\eps$, the component of $v$ has at most $\Theta(\Delta^4 \log n)$ vertices.
\end{proof}

\begin{corollary}\label[corollary]{crl:LCA}
There is a local computation algorithm that, when queried on any single node $v$, performs $\poly(\Delta)$ computation and answers whether $v \in I$, $v\in \Gamma(I)$, or $v\notin (I \cup \Gamma(I))$. Moreover, letting $H$ be the subgraph of $G$ induced by the subset $V\setminus (I \cup \Gamma(I))$ of vertices. Then, with probability $1-1/\poly(n)$, each connected component of $H$ has at most $O(\Delta^4 \log n)$ vertices.
\end{corollary}
\begin{proof}
Follows directly from \Cref{lem:LCA}, \Cref{lem:localAnalysis}, and \Cref{lem:shattering}.
\end{proof}

\subsection{The Complete LCA Algorithm for MIS}\label{subsec:completeLCA}
\begin{theorem}
There is a local computation algorithm that, when queried on any single node $v$, performs $\poly(\Delta) \log n$ computation and answers whether $v\in I'$ or not, where $I'$ is a maximal independent set, with probability $1-1/\poly(n)$. Furthermore, for each node $v$, the expected query complexity to answer this question is $\poly(\Delta)$, and for any $\eps>0$, with probability at least $1-\eps$, the query complexity is at most $\poly(\Delta) \log(1/\eps)$.
\end{theorem}
\begin{proof}
When queried on each node $v$, we first call the LCA of \Cref{crl:LCA} which uses $\poly(\Delta)$ queries and answers whether $v \in I$, $v\in \Gamma(I)$, or $v\notin (I \cup \Gamma(I))$. If $v \in I$ or $v\in \Gamma(I)$, we are done and we terminate. Otherwise, we perform more queries to reveal the connected component of $v$ in the subgraph $H$ induced by the subset $V\setminus (I \cup \Gamma(I))$.  From \Cref{crl:LCA}, we know that for each $\eps>0$, with probability at least $1-\eps$, the size of the connected component of $v$ is at most $\poly(\Delta) \log(1/\eps)$, and, with probability at least $1-1/\poly(n)$, this size is at most $\poly(\Delta) \log n$.

Suppose that $v\notin (I \cup \Gamma(I))$. We perform a breadth first search, using the LCA of \Cref{crl:LCA}, to reveal all the nodes in the component of $v$ in $H$. Concretely, we first call the LCA of \Cref{crl:LCA} on each of the neighbors of $v$. Then, we repeat the same procedure for all neighbors $u \notin (I \cup \Gamma(I))$, calling the LCA of all of their neighbors, and similarly repeat it on each of their neighbors that turn out to be in $V\setminus (I \cup \Gamma(I))$, and so on. This BFS will reveal all the nodes in the connected component of the subgraph induced by $V\setminus (I \cup \Gamma(I))$ that contains $v$, at the expense of $\poly(\Delta)$ queries per each node in this component. From \Cref{crl:LCA}, we know that any connected component of $H$ has at most $O(\Delta^4 \log (1/\eps))$ vertices, with probability at least $1-\eps$. Hence, the number of queries needed to produce the component (i.e., determining all of its nodes and edges) is $O(\Delta^4 \log (1/\eps))$ with probability at least $1-\eps$, and in particular  $O(\Delta^4 \log n)$ with probability at least $1-1/\poly(n)$.

Once the component of $v$ is determined, we augment $I$ to be a maximal independent set $I'$ by adding to $I$ a maximal independent set of this connected component of $H$. This can be computed in any arbitrary fixed manner (in a way that is independent of the questioned vertex $v$). For instance, we use the lexicographically-first greedy MIS, which processes the vertices of the component one by one according to an increasing ordering of their identifiers and adds to $I'$ any vertex for which we have not added any of its neighbors to $I'$ before. This part does not add any new queries to the computation as we have already revealed all of the nodes and edges of the component.
\end{proof}

Finally, we comment that by a technique of Alon et al.~\cite{alon2012LCA} one can transform the above algorithm into one with $\poly(\Delta \log n)$ query complexity that uses only $\poly(\Delta \log n)$ bits of randomness---the base observation is that $\poly(\Delta \log n)$-wise independence suffices for our analysis.

\bibliographystyle{alpha}
\bibliography{refs}

\end{document}